\documentclass[letterpaper, 10 pt, conference]{ieeeconf}  
\IEEEoverridecommandlockouts             
\usepackage{amssymb} %

\usepackage{bm}
\usepackage{enumerate}
\usepackage{commath}
\usepackage{graphicx}
\graphicspath{{Pictures/}}
\usepackage{amsmath}
\usepackage{mathtools}

\makeatletter
\let\MYcaption\@makecaption
\makeatother

\usepackage[font=footnotesize]{subcaption}

\makeatletter
\let\@makecaption\MYcaption
\makeatother

\usepackage{breqn}
\usepackage{cite}
\usepackage[table]{xcolor}
\usepackage{booktabs}
\usepackage{empheq}
\usepackage{amsfonts}

\usepackage{amsthm}
\usepackage{hyperref}
\usepackage{xcolor}
\usepackage{mathrsfs}
\usepackage{accents}
\usepackage{thmtools}
\usepackage{thm-restate}
\usepackage{float}
\usepackage{xcolor}
\usepackage[normalem]{ulem}

\usepackage{algorithm}
\usepackage{algpseudocode}
\usepackage{dblfloatfix}

\usepackage{etoolbox}
\usepackage{float}  %

\usepackage{comment}

\theoremstyle{plain}

\newtheorem{lemma}{Lemma}

\newtheorem{proposition}{Proposition}

\newtheorem*{theorem*}{Theorem}
\newtheorem{assumption*}{Assumption}

\declaretheorem[name=Theorem]{thm}

\theoremstyle{definition}

\newtheorem{remark}{Remark}
\newtheorem{definition}{Definition}

\newtheorem{problem}{Problem}
\newtheorem*{problem*}{Problem}

\newcommand{\myset}[1]{\mathcal{#1}}

\DeclareMathOperator*{\argmin}{argmin}

\title{ \LARGE \bf
    Computationally Efficient Safe Control of Linear Systems under \\  Severe Sensor Attacks
}

\author{Xiao Tan, Pio Ong, Paulo Tabuada,  and Aaron D. Ames%
\thanks{This work is supported by TII under project \#A6847.}
\thanks{Xiao Tan, Pio Ong, and Aaron D. Ames are with the the Department of Mechanical and Civil Engineering, California Institute of Technology, Pasadena, CA 91125, USA (Email: {\tt\small xiaotan, pioong, ames@caltech.edu}).}    
\thanks{Paulo Tabuada is with the Department of Electrical and Computer Engineering at University of California, Los Angeles, CA 90095, USA (Email: {\tt\small  tabuada@ucla.edu}).} 
}

\begin{document}

\maketitle
\thispagestyle{plain}
\pagestyle{plain}

\begin{abstract}

Cyber-physical systems are prone to sensor attacks that can compromise safety. 
A common approach to synthesizing controllers robust to sensor attacks is secure state reconstruction (SSR)---but this is computationally expensive, hindering real-time control. 
In this paper, we take a safety-critical perspective on mitigating severe sensor attacks, leading to a computationally efficient solution. 
Namely, we design feedback controllers that ensure system safety by directly computing control actions from past input-output data. Instead of fully solving the SSR problem, we use conservative bounds on a control barrier function (CBF) condition, which we obtain by extending the recent eigendecomposition-based SSR approach to severe sensor attack settings. 
Additionally, we present an extended approach that solves a smaller-scale subproblem of the SSR problem, taking on some computational burden to mitigate the conservatism in the main approach. Numerical comparisons confirm that the traditional SSR approaches suffer from combinatorial issues, while our approach achieves safety guarantees with greater computational efficiency.

  \end{abstract}

\section{Introduction}

Cyber-physical systems (CPS) integrate computational elements with physical processes in a wide range of applications, including vehicles, transportation systems, and power grids. 
CPS systems rely heavily on sensor data for decision-making and feedback control, making them prime targets for attacks. These attacks aim to deceive the system by manipulating sensor readings, potentially leading to the system taking unsafe actions that are critical to overall system integrity, as reported in \cite{amin2010stealthy,liu2011false,kwon2013security}. For instance, \cite{shoukry2013non} demonstrates a real-world experiment where a vehicle's anti-braking system is maliciously triggered by spoofing onboard encoder signals. In light of these threats, there is a pressing need in computationally efficient, safeguarding controllers that can mitigate possibly unsafe controls in real time.

Guaranteeing safety under sensor attacks is challenging because safeguarding controllers have to rely on unreliable or even deliberately 
misleading sensor measurements.  The control barrier function (CBF) \cite{Ames2017}  framework has recently gained popularity for its safety filter formulation, providing Lyapunov-like necessary and sufficient conditions for forward invariance of the safe set. Based on the CBF approach, there are research efforts on obtaining safety guarantees in the presence of imperfect sensors. For example, the measurement-robust CBF \cite{dean2021guaranteeing} is proposed to deal with bounded measurement uncertainty; Belief CBF \cite{vahs2023belief} considers  measurement uncertainty that follows a probabilistic distribution; Risk-adverse CBF \cite{singletary2022safe} is also proposed for handling obstacles with uncertain positions in the configuration space. In the sensor attack scenarios, existing researches either assume that the attack signal is bounded \cite{lin2022plug}, stochastic \cite{zhang2022safe}, or performing a certain attacking strategy\cite{arnstrom2024stealthy}. CBFs have also been used to address privacy concerns in CPS systems\cite{zhong2023secure}. All above approaches, however, can not deal with intelligent sensor attacks that are designed to be deceptive and not conforming to any aforementioned patterns. 

To safeguard the system from deceptive sensor attacks, most existing approaches focuses on reconstructing the system state. This allows for applying a safety filter design based on the estimated state as in the attack-free case. Hereafter we refer to this approach as the two-stage approach. However, as shown in \cite{mao2022computational}, reconstructing the system state is intrinsically an NP-hard problem. Brute-force approach as in \cite{chong2015observability,tan2024safety} only works well for small-size problems since the computational complexity
grows combinatorially as the number of sensors grows. This becomes a bottleneck issue when real-time feedback is needed. Computationally efficient secure state reconstruction (SSR) algorithms exist, to name a few, the convex relaxation \cite{yong2016robust},  satisfiability modulo
theory techniques \cite{shoukry2018smc}, and specialized observers \cite{shoukry2015event}. In particular, an efficient polynomial-time SSR algorithm was recently proposed in \cite{mao2022computational} that decomposes the original problem into a set of sub-problems that are much easier to solve.  An outstanding assumption for these efficient SSR algorithms is that, the true system state can be uniquely determined from the corrupted input-output data, implying that the attack is sparse. 
On the other hand, in severe sensor attack scenarios,
the plausible state estimate is possibly not unique~\cite{tan2024safety}, hindering the application of above algorithms.

\begin{figure*}[b]
    \begin{equation} \label{eq:matrices} \tag{7}
    \widetilde{Y}_i := \begin{bsmallmatrix}
        y_i(0) \\
        y_i(1) \\
        \vdots \\
        y_i(t)
    \end{bsmallmatrix},
    E_i := \begin{bsmallmatrix}
        e_i(0) \\
        e_i(1) \\
        \vdots \\
        e_i(t)
    \end{bsmallmatrix},
    U := \begin{bsmallmatrix}
        u(0) \\
        u(1) \\
        \vdots\\
        u(t-1) \\
        0
    \end{bsmallmatrix},
    \mathcal{O}_i := \begin{bsmallmatrix}
        C_i \\
        C_i A \\
        \vdots \\
        C_i A^{t}
    \end{bsmallmatrix},  F_i := \begin{bsmallmatrix}
        0 & 0 & \cdots &  0 & 0 \\
        C_iB & 0 & \cdots & 0 & 0  \\
        \vdots & & \ddots & & 0 \\
        C_iA^{t-1}B &  C_iA^{t-2}B & \cdots & C_i B & 0
    \end{bsmallmatrix},     
    \begin{aligned}
    Y_i & : = \widetilde{Y}_i - F_i U \\
      &  \text{for } i \in [p].
    \end{aligned}
\end{equation}
\end{figure*}

In this paper, we revisit the safe control problem from our previous work \cite{tan2024safety} and consider system safety for linear systems under severe sensor attacks. In continuation of the theoretical development in \cite{tan2024safety}, our focus is on the computationally efficiency of control designs. Instead of adopting the two-stage approach, we propose to directly compute the safe control condition in the CBF framework, without explicitly recovering the state. Our contributions are summarized as follows:
\begin{enumerate}
    \item We extend the decomposition-based SSR algorithm in \cite{mao2022computational} to address severe sensor attacks. We show that the extended algorithm still exhibits exponential computational complexity in such cases.
    \item Building on the extended algorithm, we develop an algorithm to compute a bound to the CBF condition. This bound enables  an efficient computation of control inputs, eliminating the high computational complexity while sufficiently ensuring system safety. Additionally,
    we discuss its conservatism and introduce another algorithm that mitigates it.
    \item We implement three approaches: a brute-force SSR, an extended decomposition-based SSR, and our proposed computationally efficient algorithm. We empirically compare them in terms of computational time and optimality gap, demonstrating the explicit trade-off between efficiency and conservatism.
\end{enumerate}
Ultimately, this paper presents an online safe control computation scheme, for linear systems, that is computationally efficient even in the presence of severe sensor attacks. 

\textbf{Notation:} For $ w\in \mathbb{N}$, define $[w]:=\{1,2,\ldots,w\}$. The cardinality of a set $\myset{I}$ is denoted by $|\myset{I}|$. Given a $w\in \mathbb{N}$, a $k$-combination from $[w]$ is a subset of $[w]$ with cardinality $k$. Denote by $\mathbb{C}_{w}^k$ the set of all $k$-combinations from $[w]$. For a matrix $C\in \mathbb{R}^{w\times n}$ and an index set $\Gamma\subseteq[w]$, denote by $C_{\Gamma}$ the matrix obtained from $C$ by removing all the rows with indices not in $\Gamma$. For a square matrix $A\in \mathbb{R}^{n\times n}, \textup{sp}(A)$ represents the set of all eigenvalues of $A$. For any two sets  $\mathcal{S}_1, \mathcal{S}_2$, the Cartesian product $\mathcal{S}_1\times  \mathcal{S}_2$ is given by $ \{ (s_1, s_2): s_1 \in \myset{S}_1, s_2\in \myset{S}_2\}$.

\section{Preliminaries and problem formulation}

Consider a discrete-time linear system under sensor attacks
\begin{equation}\label{eq:system}
    \begin{aligned}
    x(\tau+1) & = Ax(\tau) + Bu(\tau),\\
     y(\tau)  & = Cx(\tau) + e(\tau), 
     \end{aligned}
\end{equation}
where $x\in \mathbb{R}^n$, $u\in \mathbb{R}^m$, $y\in \mathbb{R}^p$, and $e\in\mathbb{R}^p$ are the state, the input, the measurement signals, and the sensor attack signals, respectively. Each entry $y_i(\tau)$ corresponds to 
one sensor measurement at time $\tau$, and we view the sensor $i$ as being under attack at time $\tau$ if $e_i(\tau)$, the corresponding $i$th entry of the attack signal $e(\tau)$, is nonzero. Throughout the paper we assume that the attacker has full knowledge of the system, including the states, the dynamics~\eqref{eq:system}, and our defense strategy. The only restriction is that it can attack at most $s$ sensors, with its choices being unknown but fixed for the entire duration. Mathematically, the index set of the unknown attacked sensors is given by $ \Gamma_A = \{i \in [p]: \exists  \tau \geq 0, ~  e_i(\tau) \neq 0\}$, and we assume:
\begin{equation} \label{eq:sensor attack assumption}
    |\Gamma_A|\leq s, s\in \mathbb{N}.
\end{equation} 

Furthermore, we consider a polytopic safety constraint
defined by a vector linear function $h: \mathbb{R}^n \to \mathbb{R}^l$ as 
\begin{equation} \label{eq:safety_requirement}
    \myset{C} = \{x\in \mathbb{R}^n:  h(x):= Hx + g \geq 0 \}
\end{equation}
with $H\in\mathbb{R}^{l\times n}$ and $g\in\mathbb{R}^l$. Our safety requirement is to constrain the system trajectory within the safe set~$\mathcal{C}$ in the presence of attacks. More discussions on the problem setting can be found in \cite{tan2024safety}.

\subsection{Discrete-time safety filter}
One approach for enforcing safety constraints is the control barrier function (CBF) framework \cite{Ames2017, agrawal2017discrete}. Given a nominal control signal $t\mapsto u_{\textup{nom}}(t)$, the CBF safety filter modifies the signal with: 
\begin{align}
        & u(t)  = \argmin_u \| u - u_{\textup{nom}}(t) \|   \label{eq:qp_control_cost}\\
       \text{s. t. } &    H(Ax + Bu) + g  \geq (1-\gamma) (Hx + g), \forall x \in \myset{X}_t^{t}. \label{eq:qp_control_constraint}
\end{align}
The filter effectively picks an input $u$ with the least deviation, pointwise, from the nominal one, in order to satisfy the CBF condition in \eqref{eq:qp_control_constraint}. From CBF theory, once the input $u(t)$ satisfies the CBF condition for the true system state $x=x_{\text{true}}(t)$ at each time $t$, the resulting state trajectory will be contained in $\mathcal{C}$.
However, we do not know the system state exactly due to the presence of the attack signal $e$, so we enforce the CBF condition for all \textit{plausible states} $x$ from the set $\myset{X}_t^{t}$, which we describe in the following subsection.

\subsection{Secure state reconstruction}

We summarize some existing results on the secure state reconstruction problem. At time $t$, we collect the input-output data $\mathcal{D}_t = (u(0),u(1), \dots, u(t-1),  y(0),y(1), \dots, y(t))\in \mathbb{R}^{mt + p(t+1)}$ after receiving the latest measurement and before choosing a control input. We only consider the case that enough measurements are collected, namely, $t\geq n$. We refer to a state $x_0$ as a \emph{plausible initial state} if there exists a trajectory $x$ of system \eqref{eq:system}  that starts at $x_0$, i.e., $x(0) = x_0$, and is consistent with the input-output data $\mathcal{D}_t$, accounting for potential attacking signals $e$ under the attack model~\eqref{eq:sensor attack assumption}. A plausible state at time $t$ can be similarly defined. The secure state reconstruction (SSR) problem thus aims to recover all plausible system states at time $t$. It has been established that any plausible initial state satisfies the following equality:
\begin{equation}\label{eq:system_compact}
    Y_i  = \mathcal{O}_i x(0)+ E_i, \ i\in [p]
\end{equation}
where the matrices are defined in \eqref{eq:matrices}.

\setcounter{equation}{7}

The set of all plausible initial states $ \myset{X}_t^0$ at time $t$ has been characterized in \cite{tan2024safety}, as summarized in Lemma~\ref{lem:X0}.
\begin{lemma} \label{lem:X0}
    The set of plausible initial states $\myset{X}_t^0$ fulfills:
    \begin{equation} \label{eq:X0set}
        \myset{X}_t^0 = \bigcup_{\Gamma = \{i_1, \ldots, i_{p-s}\} \in \mathbb{C}_{p}^{p-s}} 
 \myset{X}_{t}^{0,\Gamma} 
    \end{equation}
  where   $ \myset{X}_{t}^{0,\Gamma} = \{x\in \mathbb{R}^n: \mathcal{O}_{i}x = Y_{i} \textup{ for } i \in \Gamma \}$.
\end{lemma}
The result takes advantage of the fact that 
the attack vector $E_i$ from the $i$th sensor must be identically zero for at least $p-s$ sensors. The set $\myset{X}_{t}^{0}$ collects different plausible initial states, each of which solves \eqref{eq:system_compact} for $p-s$ sensors that are assumed to be attack-free.

From plausible initial states, we can leverage system dynamics and past input sequence to derive the set of all plausible states $\myset{X}_t^{t}$  at the current time $t$:
\begin{multline}
    \label{eq:X_t}
    \myset{X}_t^{t} = A^t( \myset{X}_{t}^{0}) + A^{t-1}Bu(0) +  \ldots + Bu(t-1)
\end{multline} 
The cardinality of the set $ \myset{X}_t^t$ is the same as that of the set $\myset{X}_t^0$, which turns out to be closely related to a system-level property termed sparse observability, defined below.
\begin{definition}[$k$-sparse observability]
    System \eqref{eq:system} is \emph{$k$-sparse observable} if $(A, C_{\Gamma})$ is observable for every  $\Gamma\in \mathbb{C}_{p}^{p-k}$.
\end{definition}

\begin{remark}
    
    In practice, it is unnecessary to incorporate all the collected input-output data for state reconstruction at time $t$. Instead, we can adopt a receding-horizon approach---using data from last $l\geq n$ time steps to reconstruct the plausible state $ \myset{X}_t^{t-l}$ and forward propagate this set to obtain $ \myset{X}_t^{t}$. This strategy reduces forward-propagation error and ensures constant memory usage and computation time.

\end{remark}

\begin{lemma} [\hspace{-0.03mm}\cite{tan2024safety,fawzi2014secure}] \label{lem:CDC paper results}
The following statements hold true. 
    \begin{enumerate}
        \item When the system \eqref{eq:system} is not $s$-sparse observable, $\myset{X}_t^0 $ can be an infinite set.
        
        \item When the system \eqref{eq:system} is $s$-sparse observable, $\myset{X}_t^0 $ is finite.
        
        \item When the system \eqref{eq:system} is $2s$-sparse observable, $\myset{X}_t^0 $ contains only the true initial state.

    \end{enumerate}
\end{lemma}

The computation of the set of plausible initial states is difficult.
In fact, it has been shown to be NP-hard in general \cite{mao2022computational}.  Equation \eqref{eq:X0set} provides one approach to compute the plausible set, which however requires an enumeration of all possible attacked sensors. This poses a heavy computational burden and becomes a bottleneck for online implementation as the number of sensors increases. Under certain conditions, it is possible to perform SSR in polynomial time as discussed in \cite{mao2022computational}. We will elaborate on this algorithm, show its limitations, and proposes our extensions in following sections.

\begin{problem}[Safety]\label{problem:safety}
    Given input-output data $\mathcal{D}_t$ collected from a discrete-time linear system \eqref{eq:system}, efficiently compute a safe control input satisfying the CBF condition in \eqref{eq:qp_control_constraint}. 
\end{problem}

\begin{problem}[Minimally invasive modification] \label{problem:min_invasive}
    Further find conditions under which the obtained control is minimally invasive with respect to the cost function in \eqref{eq:qp_control_cost}. 
\end{problem}

\section{SSR by decomposition is less efficient under severe sensor attacks} \label{sec:SSR-by-decomposition}

In this section, we review the decomposition-based SSR algorithm from \cite{mao2022computational},  discuss its limitations, and propose an extension to deal with severe sensor attacks. These results lays the necessary groundwork for our computationally efficient safe control design in Section \ref{sec:proposed design}.

\subsection{Generalized eigenspace decomposition of the system}

The core idea to reducing computational complexity is to decompose the linear system in~\eqref{eq:system} into subsystems, with each one corresponding to a generalized eigenspace.
Denoting all distinct $r\leq n$ eigenvalues of the matrix $A$ by $\{\lambda_j\}_{j\in[r]}$, the generalized eigenspace corresponding to each $\lambda_j$ is  $V^j:=\{v\in \mathbb{C}^{n}: (A-\lambda_j I)^{\alpha(\lambda_j)} v = 0 \}$, where $\alpha(\lambda_j)$ is the corresponding algebraic multiplicity. Then, there exist projection matrices\footnote{Properties and constructions of projection matrices are in the Appendix.} $P_j: \mathbb{R}^n \to V^j$ for each generalized eigenspace, such that $x=\sum_{j\in[r]}P_jx$. In what follows, we use  ${i\in [p]}$, ${j\in [r]}$, and ${k\in [l]}$ to iterate over different sensors, generalized eigenspaces, and constraints, respectively. We also use the notation $\mathrm{x}_j$ for a variable that belongs to the subspace $V^j$ for $ j\in [r]$.

The decomposition of the measurement $y\in\mathbb{R}^p$ is also possible. Particularly, the result in \cite[Lemma 1]{mao2022computational} shows that for any $i\in [p]$, the vector spaces $ \{\mathcal{O}_i(V^j)\}_{j\in [r]}$ form a basis for the vector space $\mathcal{O}_i(\mathbb{R}^n)$. As such, any $Y_i\in \mathcal{O}_i(\mathbb{R}^n)$ decomposes uniquely into $\{Y_i^j\}_{j\in[r]}$ such that  $Y_i^j \in \mathcal{O}_i(V^j)$ and $Y_i = \sum_{j\in[r]} Y_i^j$. We can construct projection matrices $\tilde{P}_i^j: \mathcal{O}_i(\mathbb{R}^n) \to \mathcal{O}_i(V^j)$ for all $i\in [p],~j\in [r]$, such that $Y_i^j = \tilde P_i^j Y_i$. Therefore, we have
\begin{equation}\label{eq:projection}
    \begin{aligned}
        \forall x\in \mathbb{R}^n,\quad \sum_{j\in [r]} P_j x = x \text{ and} \sum_{j\in[r]} \tilde{P}_i^j  \mathcal{O}_i x =  \mathcal{O}_i x.
    \end{aligned}
\end{equation}
Note that since the projection matrices are only related to system matrices $(A, C)$, we can compute them offline.

\subsection{Sub-SSR problems}
The eigen-decomposition produces subsystems that correspond to the subspaces $\{V^j\}_{j\in [r]}$ as: 
\begin{equation} \label{eq:subsystems}
    \begin{aligned}
    \mathrm{x}_j(\tau+1) & = P_jA \mathrm{x}_j(\tau) + P_j Bu(\tau),\\
     y_i^j(\tau)  & = C_iP_j x(\tau) + e_i^j(\tau), \ i\in [p],
     \end{aligned}
\end{equation}
where $y_i^j$ and $e_i^j$ are the $\tau$-th element in the vectors $Y_i^j$ and $E_i^j = \tilde{P}_i^j E_i$, respectively. A compact reformulation is:
\begin{equation} \label{eq:subsystems_compact}
    Y_i^j = \mathcal{O}_i^j \mathrm{x}_j(0) + E_i^j,~i\in [p],
\end{equation}
with $\mathcal{O}_i^j = \tilde{P}_i^j \mathcal{O}_i$. Note
the attack vector $E_i^j$ is nonzero only if the sensor $i$ is attacked. The sub-SSR problem seeks to find all plausible initial states for this subsystem given input-output data $\mathcal{D}_t^j = (U,  \{Y_i^j \}_{ i\in [p]})$.
Related to this is the notion of eigenvalue observability. 

\begin{definition}[Eigenvalue observability]
    We say an eigenvalue $\lambda\in \textup{sp}(A)$ is \emph{observable} with respect to sensor $i\in [p]$, if $\textup{Rank}\begin{bsmallmatrix}
        A-\lambda I \\
        C_i
    \end{bsmallmatrix} = n$.
\end{definition}

\begin{definition}($k$-eigenvalue observability) \label{def:eigenvalue_observability}
System~\eqref{eq:system} is \emph{$k$-eigenvalue observable} if each eigenvalue of $A$ is observable with respect to at least $k+1$ sensors.
\end{definition}

\begin{lemma}[Eigenvalue observability and sparse observability,\cite{mao2022computational}] \label{lem:eig and sparse observable}
    If system~\eqref{eq:system} is $k$-eigenvalue observable, then it is $k$-sparse observable. If system~\eqref{eq:system} is $k$-sparse observable and the geometric multiplicity is one for every eigenvalue of the matrix $A$, then the system is $k$-eigenvalue observable.
\end{lemma}

Much like the $k$-sparse observability, $k$-eigenvalue observability is an indicator of sensor redundancy in each generalized eigenspace. The geometric multiplicity assumption  mitigates the NP-hardness associated with general SSR problems, and is reasonable for many applications. To this end, \cite{mao2022computational} proposes an efficient SSR algorithm that avoids the need to consider combinatorial subsets of all sensors under a $2s$-eigenvalue observability assumption. We summarize the results from \cite{mao2022computational} in the following lemma.

\begin{lemma} [\hspace{-0.03mm}\cite{mao2022computational}] \label{lem:computation paper}
The following statements hold. 
    \begin{enumerate}
        \item The system~\eqref{eq:system} is $k$-sparse ($k$-eigenvalue) observable if and only if each subsystem~\eqref{eq:subsystems} is $k$-sparse ($k$-eigenvalue) observable.

         \item For any $x\in \mathbb{R}^n$, let $\mathrm{x}_j := P_j x$. We have the following equivalence:
         $\mathcal{O}_i x = Y_i \Leftrightarrow \mathcal{O}_i^j \mathrm{x}_j = Y_i^j, \forall j\in [r].$ 
         
        \item When the system \eqref{eq:system} is $2s$-sparse observable, both the original SSR problem~\eqref{eq:system_compact} and the sub-SSR problems~\eqref{eq:subsystems_compact} have unique solutions, denoted by $x^*$ and $\{\mathrm{x}_j^*\}_{j\in[r]}$, respectively. These solutions satisfy $x^* = \sum_{j\in [r]} \mathrm{x}_j^*$.

\item If the system~\eqref{eq:system} is $2s$-eigenvalue observable, then the unique solutions $\{\mathrm{x}_j^*\}_{j\in[r]}$ to the sub-SSR problems can be obtained by majority votes.

    \end{enumerate}
\end{lemma}
Along with these results, Algorithm~\ref{alg:SSR_by_decomposition} from~\cite{mao2022computational} is proposed to solve the SSR problem for $2s$-eigenvalue observable systems. 
The algorithm starts by collecting the reconstructed states $\{ \mathrm{x}_{j,i} \}$ on each subspace $V^j$ from each sensor $i$ for which the eigenvalue $\lambda_j$ is observable (Algorithm \ref{alg:pre-processing}). Then, it finds the substate $\mathrm{x}_j$, for each subsystem~$j$, that receives the most votes. The plausible initial state is obtained as the summation of $\mathrm{x}_j$ over all subspaces. 
The correctness of this algorithm can be easily established from Lemma~\ref{lem:computation paper} Claims 3 and 4. Unfortunately, this algorithm is not applicable under \emph{severe sensor attack} where  
the $2s$-sparse observability property does not hold, let alone the more stringent $2s$-eigenvalue observability condition.

\begin{algorithm}[h]
\caption{Solve SSR problem by decomposition for  $2s$-eigenvalue observable systems (Proposed in \cite{mao2022computational})}
\label{alg:SSR_by_decomposition}
\begin{algorithmic}[1]
\Require{system dynamics \eqref{eq:system}, input-output data $\mathcal{D}_t$, and the maximal number of attacks $s$}
\Ensure{ the set of plausible initial states  $\myset{X}_t^0$}
\State $\{\mathrm{x}_{j,i}\} \leftarrow$ execute \texttt{Pre-processing} (Alg.~\ref{alg:pre-processing})
\State for each $j\in [r]$, select $\mathrm{x}_j$ with the most votes among all $\mathrm{x}_{j,i}$
\State \Return $\myset{X}_t^0 = \{x\}$ with $x = \mathrm{x}_1 + \mathrm{x}_2 +\ldots + \mathrm{x}_r$
\end{algorithmic}
\end{algorithm}

\begin{algorithm}[h]
\caption{Pre-processing}
\label{alg:pre-processing}
\begin{algorithmic}[1]
\Require{system dynamics in \eqref{eq:system}, input-output data $\mathcal{D}_t$, and the maximal number of attacks $s$}
\For{each $(j,i) \in [r]\times [p]$}
\If{$\lambda_j$ is observable w.r.t. sensor $i$ }
\State  construct projection matrices $P_j, \tilde{P}_i^j$
\State construct sub-SSR problem \eqref{eq:subsystems_compact}
\State store $\mathrm{x}_{j,i}$ $\leftarrow$  solve $\mathrm{x}_j$ from \eqref{eq:subsystems_compact} assuming $E_i^j = 0$
\EndIf
\EndFor
\State \Return $\{ \mathrm{x}_{j,i}\}$, a collection of subspace-sensor indexed states $\mathrm{x}_{j,i}$ 
\end{algorithmic}
\end{algorithm}

\subsection{Proposed extension to address severe sensor attacks}
In what follows, we propose an extension to Algorithm~\ref{alg:SSR_by_decomposition} for the severe sensor attack scenario, where the system is only $q$-eigenvalue observable with $s\leq q \leq 2s$.
In this setting, there can be more attacked sensors than the intact ones in some subspace $V^j$. Thus,
the projected state $\mathrm{x}_j$ with the most votes from Step 2 of Algorithm~\ref{alg:SSR_by_decomposition} may be a fake state set by the attacker, making the algorithm incorrect.

We propose Algorithm~\ref{alg:SSR_by_decomposition_q_eigen} to account for the severe sensor attack scenarios. The algorithm starts by obtaining a collection of subspace-sensor indexed states $\{ \mathrm{x}_{j,i} \}$ similar to Algorithm~\ref{alg:SSR_by_decomposition}. We consider each $\mathrm{x}_j$ 
that receives votes from at least $q+1-s$ sensors as a plausible state for the subsystem~$j$ and collect them in the set $X_j$. Then, for each  state $\mathrm{x}_j$, we also record the set of sensors $\mathcal{I}_j $ that disagrees with the state $\mathrm{x}_j$, storing them $(\mathrm{x}_j, \mathcal{I}_j)$ as pairs in set $\myset{S}_j$ (Algorithm~\ref{alg:threshold_voting}). Finally, we enumerate all possible combinations in $\mathcal{S}_1\times \ldots \times \mathcal{S}_r$, and if the condition $\cup_{j\in[r]}\myset{I}_j\leq s$ holds for that combination, we sum up all corresponding $\mathrm{x}_j, j\in [r]$ and store it in the set $\myset{M}$. The rationale behind this cardinality condition is that, if more than $s$ sensors disagree with a combination, then there must be at least one intact sensor among the disagreeing sensors. A formal proof is given below.

\begin{proposition} \label{prop:algorithm_ssr_q_eigen}
    If the system~\eqref{eq:system} is $q$-eigenvalue observable with $s\leq q \leq  2s$, then for any input-output data $\mathcal{D}_t$ collected from this system, Algorithm~\ref{alg:SSR_by_decomposition_q_eigen} returns the exact set of plausible initial states as in \eqref{eq:X0set}, i.e, $\myset{M} = \myset{X}_t^0$. 
\end{proposition}
\begin{proof}
    $\myset{M} \subseteq \myset{X}_t^0$: Given $x\in \myset{M}$, there exists $\{(\mathrm{x}_j,\mathcal{I}_j)\}_{j\in [r]}$ such that $x=\sum_{j\in[r]}\mathrm{x}_j$ with $\mathrm{x}_j\in V_j$. We derive $P_jx = P_j\sum_{j\in[r]}\mathrm{x}_j = \mathrm{x}_j$ using the properties of the projection matrix.
    Therefore, from the fact that each $\mathrm{x}_j\in X_j, j\in [r],$ is a solution to $\mathcal{O}_i^j \mathrm{x}_j = Y_i^j$, for $i\in [p]\setminus \cup_{j\in [r]} \mathcal{I}_j $, we can apply Lemma~\ref{lem:computation paper} Claim 2 to deduce  $\mathcal{O}_i x = Y_i$ for each sensor $i\in[p]\setminus \cup_{j\in [r]} \mathcal{I}_j$. 
    Since there are only $|\cup_{j\in [r]} \mathcal{I}_j| \leq s $ sensors disagreeing with the combination, there exists a subset  $\Gamma\subseteq [p]\setminus \cup_{j\in [r]} \mathcal{I}_j$ with cardinality $p-s$. Together with $\mathcal{O}_i x = Y_i$ for each $i\in\Gamma$, $x$ is a member of the set $\myset{X}^0_t$.

    $\myset{X}_t^0\subseteq \myset{M}$: Given $x\in\myset{X}_t^0$, there exists an index set $\Gamma$ for $|\Gamma| = p-s$ sensors such that $O_ix=Y_i$ for $i\in\Gamma$. Recall that by definition of $q$-eigenvalue observability, there are at least $q+1$ sensors in each subspace for which the eigenvalue is observable. We will show that, among those $q+1$ sensors, at least $q+1-s$ sensors are in the set $\Gamma$. If  this is not true, it implies that there exist more than $s$ sensors not in the set $\Gamma$. This yields a contradiction since $ |[p]\setminus \Gamma|  = s$. As such, 
    according to Lemma~\ref{lem:computation paper} Claim 2, for each $j\in[r]$, $\mathrm{x}_j= P_jx$ satisfies $O_i^j\mathrm{x}_j=Y_i^j$ for at least $q+1-s$ sensors.
    Therefore, each $P_jx$ is a plausible state in the $j$th subspace. Then $|\cup_{j\in[r]}I_j|\leq s$ holds because there are only $s$ sensors not from the set $\Gamma$. As such, $\sum_{j\in[r]}P_jx$ is a plausible state belonging to the set $\myset{M}$, concluding the proof. 
\end{proof}

\begin{algorithm}[h]
\caption{Solve SSR problem by decomposition for $q$-eigenvalue observable systems, $s\leq q \leq  2s$}
\label{alg:SSR_by_decomposition_q_eigen}
\begin{algorithmic}[1]
\Require{system dynamics in \eqref{eq:system}, input-output data $\mathcal{D}_t$, and the maximal number of attacks $s$}
\Ensure{ the set of plausible initial states  $\myset{X}_t^0$}
\State $\{\mathrm{x}_{j,i}\} \leftarrow$ execute \texttt{Pre-processing} (Alg.~\ref{alg:pre-processing})
\State $\myset{S}_j, j\in [r] \leftarrow $ execute  \texttt{Threshold_voting} (Alg.~\ref{alg:threshold_voting}) 
\State Initialize $ \myset{M}  = \{ \}$
\For{each $(\mathrm{x}_1, \mathcal{I}_1,\ldots, \mathrm{x}_r, \mathcal{I}_r) \in \myset{S}_1 \times \ldots \times \myset{S}_r $ }
\If{$| \cup_j \mathcal{I}_j|\leq s$}
\State  $\myset{M} \leftarrow \myset{M}  \cup\{x\}$ with $x = \mathrm{x}_1 + \mathrm{x}_2 +\ldots + \mathrm{x}_r$
\EndIf
\EndFor

\State \Return $\myset{M} $

\end{algorithmic}
\end{algorithm}

\begin{algorithm}[h]
\caption{Threshold voting}
\label{alg:threshold_voting}
\begin{algorithmic}[1]
\Require{ a collection of indexed substates $\{ \mathrm{x}_{j,i}\}$}
\For{each subsystem $j\in [r]$}
\State  $X_j \leftarrow$ collect $\mathrm{x}_{j,i}$ with at least $q+1-s$ votes.  
\State initialize $\myset{S}_j = \{ \}$
\For{each $\mathrm{x}_j$ in $X_j$}
\State $\mathcal{I}_j \leftarrow $ collect sensors that are not consistent with the substate $\mathrm{x}_j$, i.e., $\mathcal{I}_j = \{ i\in [p]: \mathcal{O}_i^{j} \mathrm{x}_j \neq Y_i^j \} $. 
\State $\myset{S}_j \leftarrow \myset{S}_j \cup \{ (\mathrm{x}_j,\myset{I}_{j}) \} $ 
\EndFor
\EndFor
\State \Return $\myset{S}_j$ for $j\in [r]$ 
\end{algorithmic} 
\end{algorithm}

Algorithm~\ref{alg:SSR_by_decomposition_q_eigen} generalizes Algorithm~\ref{alg:SSR_by_decomposition} to address severe sensor attacks, both of which leverages the concept of eigenvalue observability to efficiently solve the sub-SSR problems. However,
from  Steps $4$-$8$ of Algorithm~\ref{alg:SSR_by_decomposition_q_eigen}, it is evident that in the presence of severe sensor attacks ($q<2s$), the decomposition-based SSR algorithm 
still faces computational burden that scales exponentially with respect to the number of eigenspaces~$r$. When incorporated into two-stage approaches--reconstructing the states first, and then computing safe control inputs--state reconstruction will still be a bottleneck and hinder real-time safe control computation. 
In the following section we propose a computationally efficient, albeit conservative, approach to safe control computation.

\section{Computationally efficient safe control} \label{sec:proposed design}
In this section, we propose a computationally efficient safe control design without explicitly reconstructing the set of plausible states. Our approach directly computes a sufficient safe condition based on the control barrier function framework. The conservatism of our approach as well as mitigation methods are also discussed.

We first rewrite the CBF condition in \eqref{eq:qp_control_constraint} by using \eqref{eq:X_t} as
\begin{equation}
\begin{aligned}
    HB u  & \geq  H\left((1-\gamma)I -A\right)(A^t x  + A^{t-1}Bu(0)   \\
      & \hspace{5mm}  +  ... + Bu(t-1)) -\gamma g  \text{ for all }  x\in \myset{X}_{t}^{0}
\end{aligned}
\end{equation}
The above inequality is equivalent to 
\begin{equation}  \label{eq:equivalent_CBF_condition}
    HB u  \geq M(Y) + Q(U),
\end{equation}
where, as shorthand notations, we introduce the functions:
\begin{equation*}
    \begin{aligned}
     &  [M(Y)]_k := \max_{x \in \myset{X}_t^0} [Kx]_k, \textup{ for } k\in [l], \\ 
    & Q(U) := K_0 ( A^{t-1}Bu(0) + ... + Bu(t-1)) -\gamma g  
    \end{aligned}
\end{equation*}
with $K_0 :=H\left((1-\gamma)I -A\right)$ and $K  :=  K_0 A^t$. Here, we make it clear that, $M(Y)\in \mathbb{R}^l$ is a vector defined entry-wise, with $l\in \mathbb{N}$ being the total number state constraints from \eqref{eq:safety_requirement}. The values of $\max(\cdot)$ function are well-defined in \eqref{eq:equivalent_CBF_condition} since the set $\myset{X}_t^0$ is finite, cf. Lemma~\ref{lem:CDC paper results}. We also make explicit the dependence on $Y$ that characterizes the set $\myset{X}_t^0$ in \eqref{eq:X0set}. Given that $Q(U)$ can be computed directly from the input data, the main obstacle in expressing the CBF condition is the computation of $M(Y)$.

We propose Algorithm~\ref{alg:M_upper_bound} to efficiently compute an upper bound on $ M(Y)$ given input-output data $\mathcal{D}_t$. The algorithm first computes all the subspace-sensor indexed states $\{  \mathrm{x}_{j,i} \}$, a step similar to previous algorithms. Then, for each subspace $V^j$, we collect all states $\mathrm{x}_j$ that have at least $q+1-s$ votes into a set $X_j$, and compute entry-wise the maximal values of $Kv_j$ for $v_j \in X_j$. The upper bound on $M(Y)$ is then obtained by summing up those maximal values from different subspaces. Here, instead of ruling out combinations of $\mathrm{x}_j$ that are not plausible, we directly use the worst combination of $\mathrm{x}_j$, even if it may not be plausible, in order to avoid the exponential computation issue. A formal guarantee on this algorithm is established below.

\begin{algorithm}[h]
\caption{Compute an upper bound on $M(Y)$ for $q$-eigenvalue observable systems, $s\leq q < 2s$}
\label{alg:M_upper_bound}
\begin{algorithmic}[1]
\Require{system dynamics in \eqref{eq:system}, input-output data $\mathcal{D}_t$, and the maximal number of attacks $s$}
\Ensure{$\overline{M}$: an upper bound of $  M(Y_t) $}
\State $\{\mathrm{x}_{j,i}\} \leftarrow$ execute \texttt{Pre-processing} (Alg.~\ref{alg:pre-processing})
\For{each subsystem $j\in [r]$}
\State  $X_j \leftarrow$ collect $\mathrm{x}_{j,i}$ with at least $q+1-s$ votes.
\State for $k\in [l]$, compute $ [M_j]_k  \leftarrow  \max_{\mathrm{x}_j  \in X_j } [K\mathrm{x}_j]_k $
\EndFor
\State \Return $\overline{M} = \sum_{j} M_j$

\end{algorithmic}
\end{algorithm}

\begin{proposition} \label{prop:M_upper}
    Consider a $q$-eigenvalue observable system~\eqref{eq:system} with $s\leq q \leq  2s$. Given input-output data $\mathcal{D}_t$ collected from this system, let $\overline{M} $ be the output from Algorithm~\ref{alg:M_upper_bound}. Then $\overline{M} \geq M(Y)$ holds.
\end{proposition}

Proposition \ref{prop:M_upper} will be proven later as part of Proposition~\ref{prop:M_upper_less_conservative}. Based on Proposition \ref{prop:M_upper}, we obtain a conservative CBF condition for the system input $u$ as:
\begin{equation} \label{eq:conservative_CBF_condition}
    HBu \geq \overline{M}(Y) + Q(U).
\end{equation}
Here we make explicit the dependency of the output of Algorithm~\ref{alg:M_upper_bound} on the input-output data. We also note importantly that, unlike previous algorithms, $\overline{M}(Y)$ does not require an enumeration of plausible states. At the same time, the satisfaction of the constraint~\eqref{eq:conservative_CBF_condition} leads to safety.

\begin{thm}\label{thm:safety}
    Consider a $q$-eigenvalue observable system in \eqref{eq:system} with $s \leq q \leq 2s$. Any input sequence $\{u(\tau)\}_{\tau\geq n}$ fulfilling the constraint~\eqref{eq:conservative_CBF_condition} for all $\tau\geq n$     
    produces a state trajectory $\tau\mapsto x(\tau)$ with following properties: 
       \begin{equation} \label{eq:safety_guarantee}
           \begin{aligned}
               & x(n-1)\in \mathcal{C} \implies  x(\tau)\in \mathcal{C},~\forall \tau\geq n, \\
               & x(n-1)\notin \mathcal{C} \implies  \lim_{\tau \to \infty} \| x(\tau) \|_{\mathcal{C}} = 0.
           \end{aligned}
       \end{equation}
\end{thm}
\begin{proof}
 Proposition~\ref{prop:M_upper} shows that any input $u$ fulfilling the condition in \eqref{eq:conservative_CBF_condition} also satisfies the CBF condition in \eqref{eq:qp_control_constraint}. The two implications follow from CBF theory \cite{agrawal2017discrete,Ames2017}.
\end{proof}

Theorem~\ref{thm:safety}, address Problem \ref{problem:safety} by dictating a condition on the feedback controls that guarantees safety. To address Problem~\ref{problem:min_invasive}, we design a safety filter that computes safeguarding control inputs via quadratic programming:
\begin{equation} \label{eq:conservative_controller}
     u(t)  = \argmin_u \| u - u_{\textup{nom}}(t) \|  \quad   \text{s. t. }    \eqref{eq:conservative_CBF_condition} \textup{ holds.}
\end{equation}
Compared to the CBF condition~\eqref{eq:qp_control_constraint}, which is based on the exact set of plausible states, the  constraint~\eqref{eq:conservative_CBF_condition} enforces a conservative safeguarding control input that deviates further from the nominal control input $u_{\textup{nom}}$. This represents a trade-off between optimality and computational efficiency in these two approaches. Next, we present an algorithm that strikes a balance between the two.

\begin{remark}
    Verifying the feasibility of the CBF conditions in \eqref{eq:equivalent_CBF_condition} and \eqref{eq:conservative_CBF_condition} can be challenging. For the CBF condition in \eqref{eq:equivalent_CBF_condition}, a thorough analysis on its feasibility was provided in our previous work \cite{tan2024safety}. The analysis therein also applies to the conservative CBF condition \eqref{eq:conservative_CBF_condition}.
\end{remark}

\subsection{Improving optimality via partial subspace combinations}

We propose Algorithm~\ref{alg:M_upper_bound_less_conservative} to calculate a tighter bound on~$M(Y)$. The algorithm considers an index subset $\Lambda\subseteq[r]$ for the generalized eigenspaces (Step 2) and aims to rule out some implausible combinations $\{\mathrm{x}_j\}_{j\in\Lambda}$ that have $\vert\cup_{j\in\Lambda}\vert>s$. In this case, the total number of combinations is limited to a smaller size. At the same time, the bound $\overline{M}_\Lambda(Y)$, i.e., the output of Algorithm~\ref{alg:M_upper_bound_less_conservative}, better approximates $M(Y)$ than $\overline{M}(Y)$ does, as guaranteed in the following result.

\begin{proposition}
\label{prop:M_upper_less_conservative}
Consider a $q$-eigenvalue observable system~\eqref{eq:system} with $s\leq q \leq  2s$. Given input-output data $\mathcal{D}_t$ collected from this system, let $\overline{M}$ and $\overline{M}_\Lambda$ be the outputs from Algorithm~\ref{alg:M_upper_bound} and Algorithm~\ref{alg:M_upper_bound_less_conservative}, respectively. Then $\overline{M} \geq \overline{M}_\Lambda\geq M(Y)$ holds.
\end{proposition}

\begin{proof}
    Consider an arbitrary $k$. 
    Let $x^*\in \myset{M}$ be the maximizer of  $\max_{x\in\myset{M}}[Kx]_k$. From the construction of $\myset{M}$ in Algorithm~\ref{alg:M_upper_bound}, there exist $\{\mathrm{x}_j^*\}_{j_\in[r]}\in X_j$, where $X_j$ is from Step 3 of Algorithm \ref{alg:M_upper_bound}, such that $x^*=\sum_{j\in[r]}\mathrm{x}_j'$ and associated $\{\myset{I}_j^*\}_{j\in[r]}$ such that $\vert\cup_{j\in[r]}\myset{I}_j^*\vert<s$. Since $\Lambda\subseteq [r]$, it is necessary that $\vert\cup_{j\in\Lambda}\myset{I}_j^*\vert<s$, which implies that $\sum_{j\in_\Lambda}\mathrm{x}_j^*$ belongs to the set $X_\Lambda$ defined in Step 3 of Algorithm~\ref{alg:M_upper_bound_less_conservative}. As a result, we have:
    \begin{align*}
     [M(Y)]_k &=   [Kx^*]_k = \sum_{j\in[r]\setminus \Lambda}[K\mathrm{x}_j^*]_k + [K\sum_{j\in\Lambda}\mathrm{x}_j^*]_k \\
     &\leq \sum_{j\in[r]\setminus\Lambda} \max_{\mathrm{x}_j  \in X_j } [K\mathrm{x}_j]_k + \max_{\mathrm{x}_\Lambda  \in X_\Lambda } [K\mathrm{x}_\Lambda]_k = \overline{M}_\Lambda.
     \end{align*}
    The first equality holds because $X_t^0=\myset{M}$ from Proposition~\ref{prop:algorithm_ssr_q_eigen}, and the inequality holds because of the memberships $\mathrm{x}_j^*\in X_j$ and $\sum_{j\in\Lambda}\mathrm{x}_j^*\in X_\Lambda$.

    Let $\mathrm{x}_\Lambda'\in X_\Lambda$  be the maximizer of $\max_{\mathrm{x}_\Lambda  \in X_\Lambda } [K\mathrm{x}_\Lambda]_k$. From the construction of $\myset{X_\Lambda}$ in Algorithm~\ref{alg:M_upper_bound_less_conservative}, there exists $\{\mathrm{x}_j'\}_{j\in\Lambda}$ with $\mathrm{x}_j'\in X_j$ such that $\mathrm{x}_\Lambda' = \sum_{j\in\Lambda} \mathrm{x}_j'$. Therefore,
     \begin{align*}
     [\overline{M}_\Lambda]_k &= \sum_{j\in[r]\setminus\Lambda} \max_{\mathrm{x}_j  \in X_j } [K\mathrm{x}_j]_k + \max_{\mathrm{x}_\Lambda  \in X_\Lambda } [K\mathrm{x}_\Lambda]_k \\
     &= \sum_{j\in[r]\setminus\Lambda} \max_{\mathrm{x}_j  \in X_j } [K\mathrm{x}_j]_k +  \sum_{j\in\Lambda}[K\mathrm{x}_j']_k\\
     &\leq \sum_{j\in[r]} \max_{\mathrm{x}_j  \in X_j } [K\mathrm{x}_j]_k = [\overline M]_k.
     \end{align*}
     Since $k$ is arbitrary, we obtain $\overline{M} \geq \overline{M}_\Lambda\geq M(Y)$.
\end{proof}

Similar safety results as in Theorem~\ref{thm:safety} can also be established using the tighter upper bound estimate $\overline{M}_\Lambda(Y)$. In addition, the ordering $\overline{M}(Y)\geq\overline{M}_\Lambda(Y)\geq M(Y)$ imposes that the resulting cost $\|u(t)-u_{\textup{nom}}(t)\|$, with $u(t)$ computed with the respective bounds in the constraint, follows the same order. We omit these results here due to space limits.

\begin{algorithm}[h]
\caption{Compute a tighter upper bound on $M(Y)$  for $q$-eigenvalue observable systems, $s\leq q \leq 2s$}
\label{alg:M_upper_bound_less_conservative}
\begin{algorithmic}[1]
\Require{system dynamics in \eqref{eq:system}, input-output data $\mathcal{D}_t$, and the maximal number of attacks $s$}
\Ensure{$\overline{M}$: an upper bound of $  M(Y_t) $}
\State Execute Steps 1 - 5 of Algorithm~\ref{alg:M_upper_bound}.
\State Select a subset $\Lambda\subseteq [r]$ of all subspaces.
\State initialize $X_{\Lambda} = \{ \}$
\For{each $(\ldots, \mathrm{x}_j, \mathcal{I}_{j},\ldots) \in \bigtimes_{j \in \Lambda }\myset{S}_{j}$}
\If{$|\cup_{j\in \Lambda} \mathcal{I}_{j}| \leq s$}
\State $X_{\Lambda} \leftarrow X_{\Lambda} \cup\{ \mathrm{x}_{\Lambda}\} $ with $\mathrm{x}_{\Lambda} = \sum_{j\in \Lambda} \mathrm{x}_j $
\EndIf
\EndFor
\State  for $k\in [l]$, compute $ [M_{\Lambda}]_k  \leftarrow  \max_{x_\Lambda  \in X_{\Lambda} } [Kx_\Lambda]_k $
\State \Return $\overline{M}_\Lambda = \sum_{j\in [r]\setminus \Lambda} M_j + M_{\Lambda}$

\end{algorithmic}
\end{algorithm}

\subsection{Discussion on attacking strategies}
In this subsection, we briefly discuss ``good'' attacking strategies from the perspective of attackers. We specifically focus on how the attacker undermines our safety guarantees, exploiting its knowledge about the system dynamics, safety constraints, as well as our defense strategy. A good attacking strategy would meet the following criterion:
\begin{enumerate}
    \item the compromised measurements from any attacked sensors must be consistent with attack-free subsystem dynamics, i.e., $Y_i^j \in \textup{Range}(\mathcal{O}_i^j)$ for the attacked sensor $i$ and a subsystem $j$. Otherwise, this attack is ignored during \texttt{Pre-processing}.
    \item the attacked sensors have to synchronize to pass the voting threshold. Otherwise, the attack will be identified during \texttt{Threshold_voting}. In fact, for each subspace $V^{j}$, the attacker can inject at most $\left\lfloor \frac{q+1}{q+1 - s} \right\rfloor - 1  =  \left\lfloor  \frac{s}{q+1 - s} \right\rfloor $ fake initial substates into the set $X_j$ defined in Step 3 of Algorithm~\ref{alg:M_upper_bound}, where $\lfloor \cdot \rfloor$ is the floor function.
\end{enumerate}

\subsection{Other improvements}

Further improvements can be made to our proposed algorithms. For many systems, the eigenvalue observability indices may not be uniformly $q$ across all eigenvalues. Suppose that the eigenvalue observability  index corresponding to eigenvalue $\lambda_j$ is $q_j, q_j > q$. Then we can replace the threshold $q+ 1 -s $ with a more stringent condition $q_j + 1 -s$ for the subspace $j$ in all relevant algorithms, while all theoretical results remain intact.

Moreover, we can remove some attacked sensors online once they have been identified. For instance, suppose that an eigenvalue $\lambda_j$ is observable with respect to a sensor $i$, yet the substate $\mathrm{x}_{j,i}$,  obtained from the measurement of sensor $i$ in the subspace $V^j$, fails to pass the threshold voting. In this case, we can conclude that the sensor $i$ is compromised and we can safely exclude its measurements from all subsystems.

\section{Simulation comparison}
We illustrate the effectiveness of our approaches through extensive simulation experiments. All the simulations are done on an Intel i7-1365U laptop in Python 3.11. \footnote{ The code is publicly available at \url{https://github.com/xiaotan-git/safety-under-sensor-attacks}.}

\subsection{Problem setup}
We run experiments using our proposed algorithms on different \textit{random} systems to compare their computation efficiency. For each system $(A,B,C)$, we generate a matrix $A$ with state dimension $n$ while making sure there are $r = n$ total generalized eigenspaces. We accomplish this by constructing $A=RJR^{-1}$ from a diagonal matrix $J\in\mathbb{R}^{n\times n}$ with distinct eigenvalues and a random (invertible) matrix $R\in\mathbb{R}^{n\times n}$ as the eigenbasis. The matrix $B$ is chosen as an identity matrix to ensure the CBF condition~\eqref{eq:qp_control_constraint} is always feasible. As for the sensors, we randomly select $s$ out of total $p$ sensors to be compromised. We ensure that the system is $q$-eigenvalue observable by constraining each eigenvalue to be observable with respect to $q+1$ random sensors. To achieve this, we take a random $C_i' \in \mathbb{R}^{1\times n}$, construct the subspace it cannot observe by stacking relevant eigenbases, then we subtract its projection on that subspace. For the safety constraint, we consider the box constraint $Hx + g \geq 0$ with $H = \begin{bsmallmatrix}
    I_{n\times n} \\
    -I_{n\times n}
\end{bsmallmatrix}$ and $g = 10\mathbf{1}_n$.

To simulate a ``good" attack, the $s$ compromised sensors will provide dynamically consistent but false measurements, corresponding to state trajectories originating from fake initial states. For all simulations below, we initialize the system with the true initial state $x_{\text{true}}$, and we use $x_{\text{fake,1}}$ and $x_{\text{fake,2}}$ as two distinct fake initial states.

\vspace{-2mm}
\subsection{Runtime comparison between different approaches}
We compare the computation times between three different algorithms for input-output data $\mathcal{D}^t$ generated using the random system matrices, random input sequence, random attacked sensors, and random initial states. These algorithms include the brute-force SSR based on \eqref{eq:X0set} \cite{tan2024safety}, our extended decomposition-based SSR (Alg.~\ref{alg:SSR_by_decomposition_q_eigen}), and the computationally efficient algorithm using an upper bound of $M(Y)$ (Alg.~\ref{alg:M_upper_bound}). 

Theoretically, the total number of combinations the brute-force approach iterates over is $\frac{p!}{(p-s)!\,s!}$. Fig.~\ref{fig:computation_complexity}(a) shows the increase in the computation time with increasing sensor number $p$ and fixed system dimensions $n, m,$ number of sensor attacks $s$, and eigenvalue observability index $q$. We also observe that the decomposition-based approaches (Alg.~\ref{alg:SSR_by_decomposition_q_eigen} and~\ref{alg:M_upper_bound}) generally performs better than the brute-force approach when $p$ increases.  On the other hand, the extended decomposition-based SSR (Alg.~\ref{alg:M_upper_bound}) requires to enumerate, in the worst case, a total number of  $(\lfloor \frac{q+1}{q+1 -s}\rfloor)^r $ combinations among $r$ subspaces. Fig.~\ref{fig:computation_complexity}(b) shows this exponential trend in computation time with respect to the number of generalized eigenspaces $r$. This shows that the proposed SSR approach is less suitable for large-dimensional systems. We also note that in all these simulated runs, the plausible states identified by the brute-force SSR and by the extended decomposition-based SSR are the same, affirming Proposition~\ref{prop:algorithm_ssr_q_eigen}.

\begin{figure}[h]
	\centering
	\begin{subfigure}[t]{0.47\linewidth}
		\includegraphics[width=\linewidth]{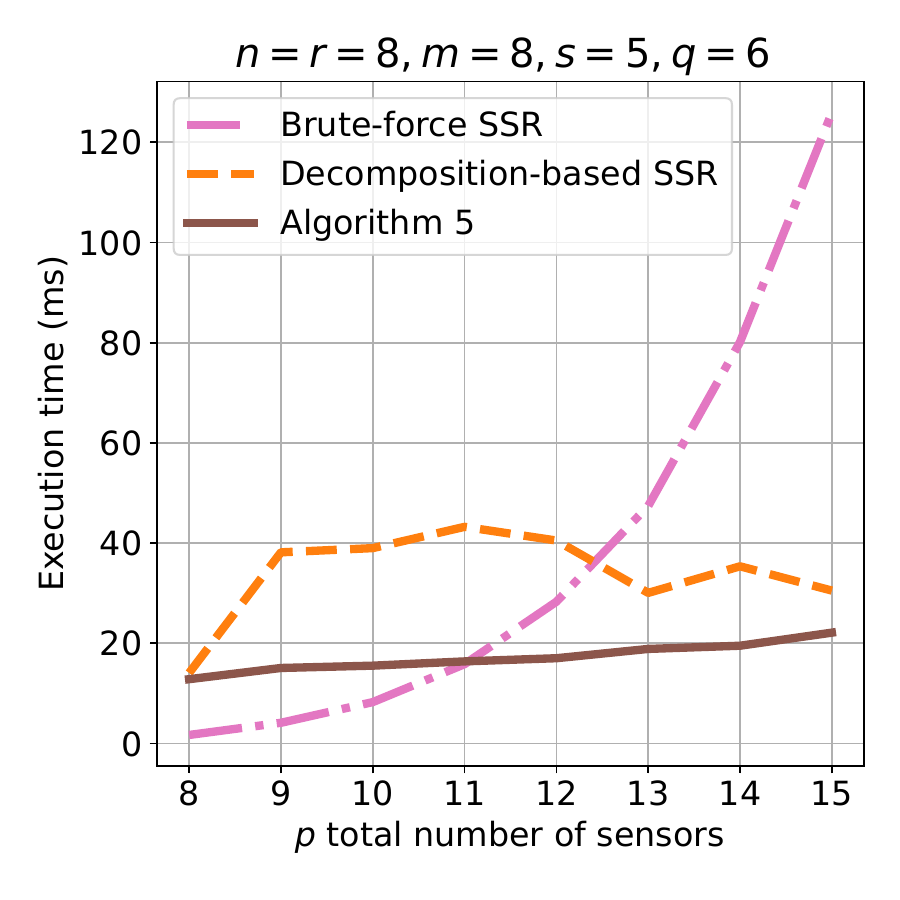}
		\caption{ Varying number of sensors. }   
	\end{subfigure} \hspace{1mm}
	\begin{subfigure}[t]{0.47\linewidth}
		\centering\includegraphics[width=\linewidth]{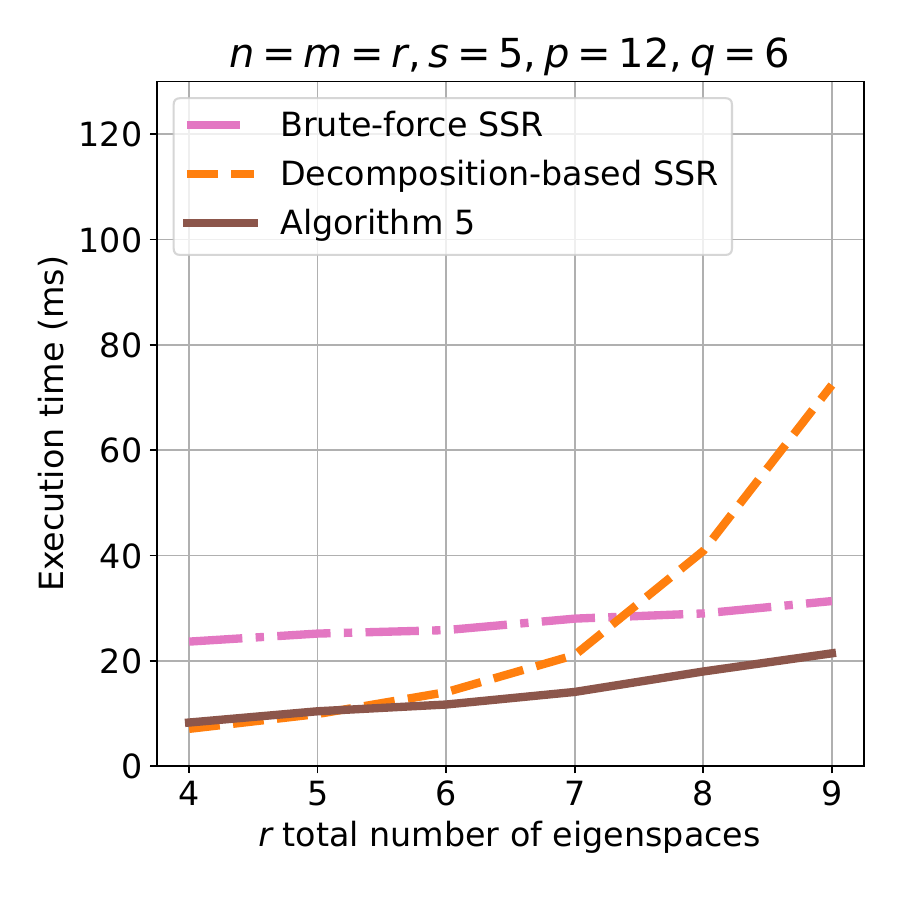}
		\caption{  Varying number of eigenspaces.}
	\end{subfigure}
	\caption{ Comparison of execution times of different algorithms. Each datapoint corresponds to an average of $100$ runs. Subfigure~(a): fixed number of attacks $s$, eigenvalue observability index $q$, and a varying number of total sensors $p$. Subfigure~(b): fixed number of sensors $p$, attacks $s$, the eigenvalue observability index $q$, and a varying Varying number of eigenspaces $r$.   }  
 \vspace{-2mm}
	\label{fig:computation_complexity}	
	\end{figure}

\subsection{Closed-loop performance comparison}
In this subsection, we show the safety assuring property of our proposed approaches. To facilitate our illustration, we choose the system matrix $A =  \tiny{\frac{1}{10}\begin{bsmallmatrix}
    8 &  4 &  0 &  0 \\
    4 & 6 & 2 & 0 \\
    0 & 2 & 5 & 3 \\
    0 & 0 & 3 & 7 \\
\end{bsmallmatrix}}$ and $B$ the identity matrix. Note that this system is inherently unstable because one eigenvalue of $A$ is greater than 1. We choose $p = 11$ random sensors such that the system has $q = 8 $ eigenvalue observability, with random $s = 5$ of them compromised. Choose $x_{\text{true}} = \mathbf{1}_{4}$, $x_{\text{fake,1}} = - \mathbf{1}_{4} $, $x_{\text{fake,2}} = 2 \mathbf{1}_{4}$, and $u_{nom}: \tau \mapsto 4(\sin(\tau), \cos(\tau),-\sin(\tau), -\cos(\tau)) $. 

The closed-loop system trajectories are presented for four different controllers: the nominal controller, the two-stage safety filters in \eqref{eq:qp_control_cost}-\eqref{eq:qp_control_constraint} with brute-force SSR and our proposed decomposition-based SSR, and our proposed efficient safety filter \eqref{eq:conservative_controller}. Fig.~\ref{fig:closed-loop traj} shows the nominal system trajectory evolves out of the safety region while all other trajectories remains safe. Note that the trajectories for the two two-stage approaches coincide, confirming Proposition~\ref{prop:algorithm_ssr_q_eigen}.
Fig.~\ref{fig:pointwise_cost} compares the step-wise costs between the three safety-filters. Consistent with our analysis, the computationally efficient approach (Alg.~\ref{alg:M_upper_bound}) has a larger but still comparable cost.

\begin{figure}
    \centering
    \includegraphics[width=0.9\linewidth]{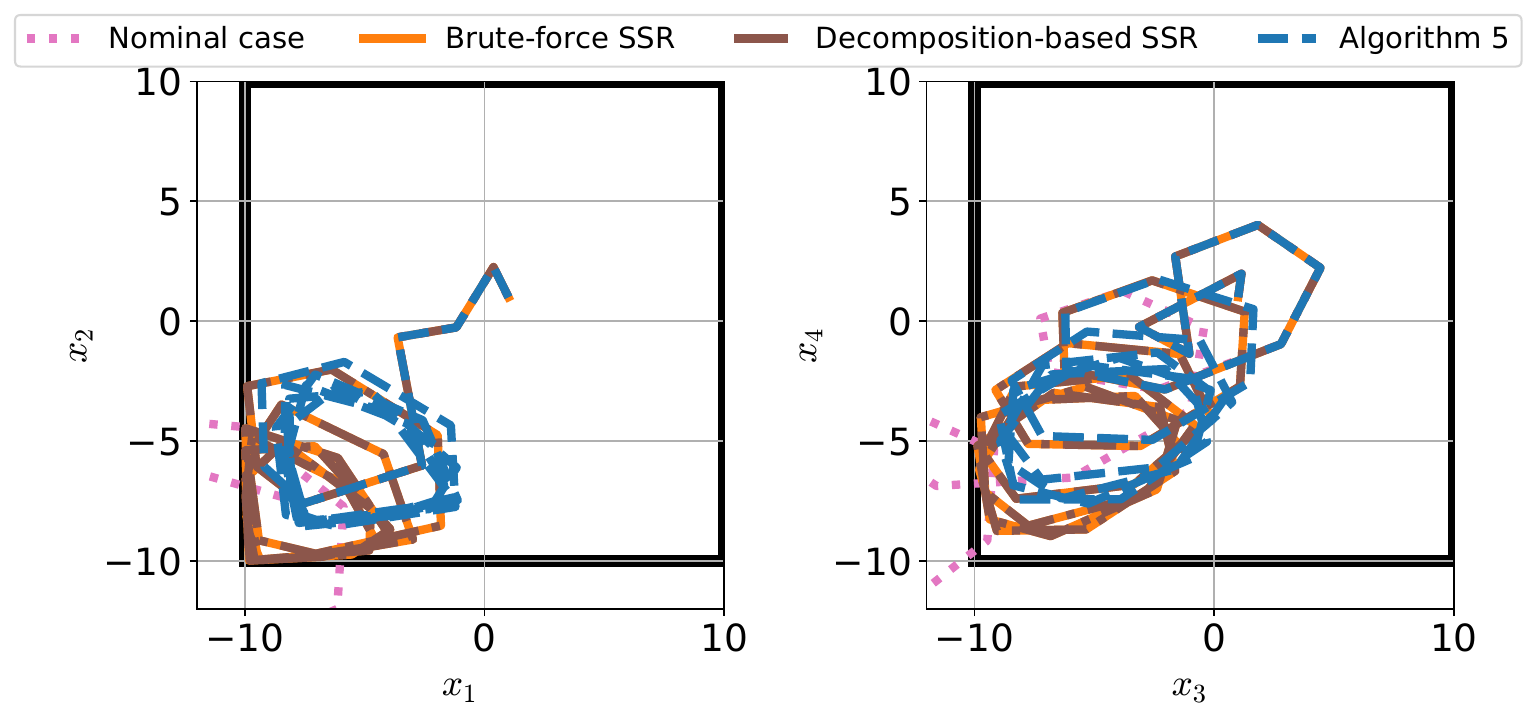}
    \vspace{-2.5mm}
    \caption{Comparison of the closed-loop system trajectories. While the nominal input derives the system out of the safety region (within the black box), all three safety filters render the system safe. Notably, the SSR-based approaches have the same safe trajectory, and the trajectory resulting from the computationally efficient approach shows a larger safety buffer. }
    \label{fig:closed-loop traj}
\end{figure}

\begin{figure}
    \centering
    \includegraphics[width=0.6\linewidth]{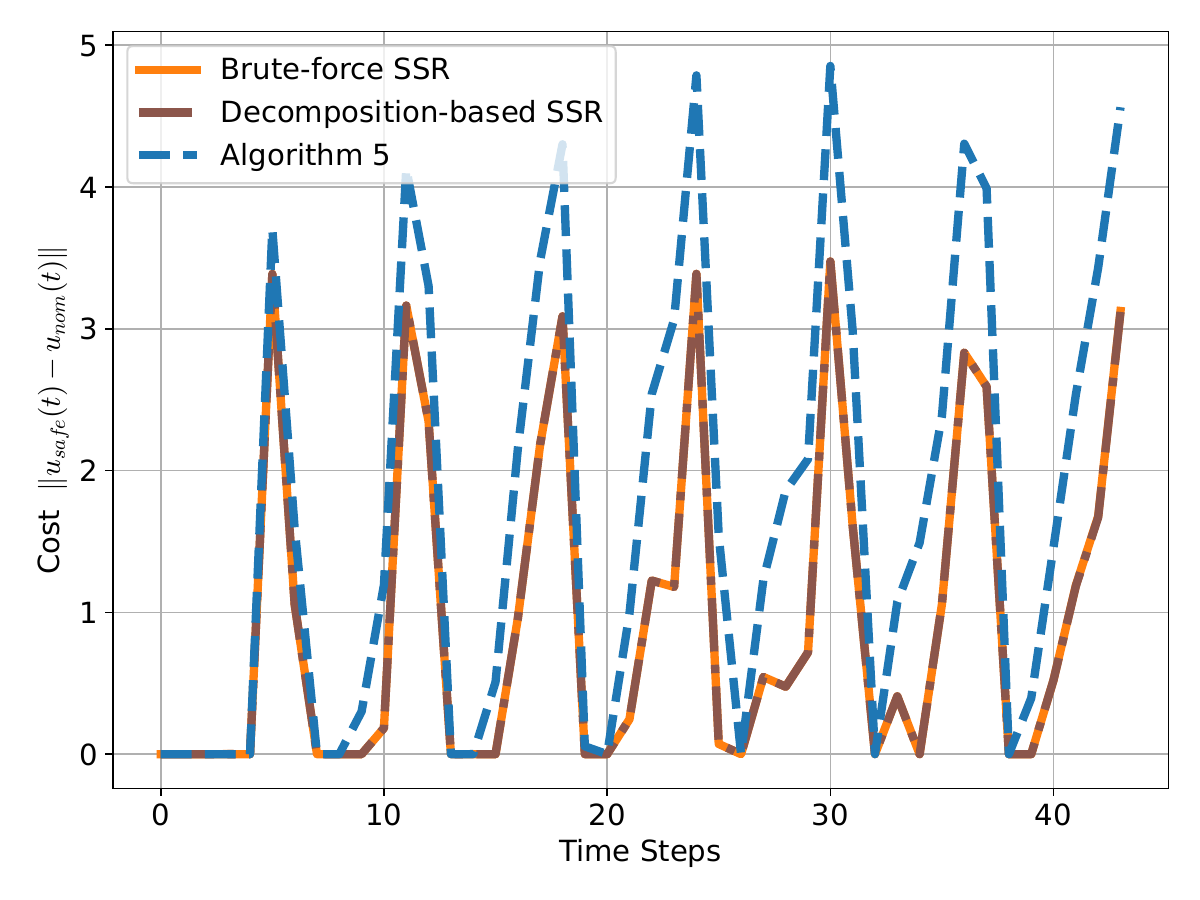}    \vspace{-2.5mm}
    \caption{Step-wise cost comparison of different approaches along the trajectory under the proposed controller \eqref{eq:conservative_controller}.}
    \label{fig:pointwise_cost}    %
    \setlength{\abovecaptionskip}{-18pt}

\end{figure}

\section{Conclusion}

In this paper, we consider safety control problem for linear systems under severe sensor attacks. Instead of using a two-stage approach that recovers the exact state, which will encounter the hard computational complexity issue, we propose an efficient approach that directly constructs the safety condition on the system input based on the CBF framework. We extend a recent eigenspace decomposition-based SSR algorithm to deal with severe sensor attacks. Its complexity, however, still grows exponentially with the number of eigenspaces. To tackle this, we propose an computationally inexpensive approach to derive an overapproximation of the exact CBF condition. We demonstrate its computational efficiency, safety assuring property, and its conservatism, empirically using randomly generated experiments. We will investigate the secure and safe control problem under measurement noise in a future work.

\section*{Appendix: Projection Matrices} \label{sec:appendix}
Consider vector spaces $\myset{S}, \myset{S}^1, \ldots, \myset{S}^r $ embedded in the Euclidean space $\mathbb{R}^m$. Suppose that the vector spaces $\{ \myset{S}^j \}_{j\in [r] }$ form a basis of the vector space $\myset{S}$. Let the matrices of basis vectors corresponding to the spaces $ \myset{S}^1, \ldots, \myset{S}^r$ be $S_1 \in   \mathbb{R}^{m\times n_1}, \ldots, S_r \in  \mathbb{R}^{m\times n_r}$, respectively.  Then, the following properties hold. (i)~$\text{Rank}(S_j) = n_j$, and $\text{dim}(\myset{S})  = \sum_{j\in [r]} n_j$. (ii) There exists a canonical projection $P_j: \myset{S} \to \myset{S}^j$ for each space $\myset{S}_j, j\in [r],$ such that (ii.1) $P_j$ is a linear map; (ii.2) $P_j s_j = s_j$ for any $s_j \in \myset{S}_j$; (ii.3) and $P_j s_k = 0$ for any $s_k \in \myset{S}_k, k\neq j$.

One such projection matrix can be computed by: 
\begin{equation*}
    P_j = [0_{m\times n_1}, \ldots, S_j, \ldots, 0_{m\times n_r}]  (S^\top S )^{-1} S^\top \quad j \in [r],
\end{equation*}
where the matrix $S: = [S_1, S_2,\ldots, S_r] \in \mathbb{R}^{m\times  \sum_{j\in [r]} n_j}$. From the first property, we know the matrix $S$ has full column rank so $(S^\top S )^{-1}$ is invertible. One can easily verify that the second property holds by observing that $P_j S = [0_{m\times n_1}, \ldots, S_j, \ldots, 0_{m\times n_r}]. $

\bibliographystyle{IEEEtran}
\bibliography{IEEEabrv,references}

\end{document}